\documentclass[11pt]{article}
\usepackage{amsfonts}
\usepackage{amssymb}
\usepackage{amsmath}
\usepackage{theorem}
\usepackage{graphicx}

\newtheorem{theo}{Theorem}[section]
{\theorembodyfont{\rmfamily}
\newtheorem{definition}[theo]{Definition}
}

\newenvironment{proof}{\textit{Proof.}}{\hfill$\square$}

{\theorembodyfont{\rmfamily}

}
{\theorembodyfont{\rmfamily}

}
{\theorembodyfont{\rmfamily}

}

\begin{document}

\title{On egalitarian values for cooperative games with level structures\footnote{This work is part of the R+D+I project grants MTM2017-87197-C3-1-P, MTM2017-87197- C3-3-P, PID2021-124030NB-C31 and PID2021-124030NB-C32, funded by MCIN/AEI/10.13039/501100011033/ and by "ERDF A way of making Europe”/EU. This research was also funded by Grupos de Referencia Competitiva ED431G-2019/01, ED431C-2020/14 and ED431C-2021/24 from the Conseller\'ia de Cultura, Educaci\'on e Universidades, Xunta de Galicia. During the completion of this work, J.~C. Gon\c calves-Dosantos is a "Margarita Salas'' postdoctoral researcher, carrying out a research stay at the CIO, {\it Centro de Investigaci\'on Operativa} of the Miguel Hern\'andez University of Elche (Spain), to which he is grateful for the welcome.
	}}
\author{J.M. Alonso-Meijide$^1$, J. Costa$^2$, \\
I. Garc\'{\i}a-Jurado$^3$, J.C. Gon\c{c}alves-Dosantos$^4$}
\date{\empty}
\maketitle

\begin{abstract}
In this paper we extend the equal division and the equal surplus division 
values for transferable utility cooperative games  to the more general setup of  transferable utility cooperative games with level structures. In the case of the equal surplus division value we propose 
three possible extensions, one of which has already been described in the literature.
We provide axiomatic characterizations of the values considered, apply them to a particular cost sharing problem and compare them in the framework of such an application.
\end{abstract}

\footnotetext[1]{Universidade de Santiago de Compostela, Grupo MODESTYA, CITMAga and Departamento de Estat\'{\i}stica, An\'alise Matem\'atica e
Optimizaci\'on, Facultade de Ciencias, Campus de Lugo, 27002 Lugo, Spain.} 
\footnotetext[2]{Universidade da Coru\~{n}a, Grupo MODES, Departamento de Matem\'aticas, Campus de Elvi\~{n}a, 15071 A Coru\~{n}a, Spain.}
\footnotetext[3]{Universidade da Coru\~{n}a, Grupo MODES, CITMAga and Departamento de Matem\'aticas, Campus de Elvi\~{n}a, 15071 A Coru\~{n}a, Spain.}
\footnotetext[4]{Universidade da Coru\~{n}a, Grupo MODES, CITIC and Departamento de Matem\'aticas, Campus de Elvi\~{n}a, 15071 A Coru\~{n}a, Spain.}


\section{Introduction}

In many practical situations, when a group of agents are faced with sharing the costs of a project they are jointly developing, they use an egalitarian approach to cost sharing. Egalitarian sharing is a clear criterion, computationally simple and easily acceptable to all as a non-conflict generating procedure. Selten (1972)
indicates that egalitarian considerations explain in a successful way
observed outcomes in experimental cooperative games.

Cooperative game theory, which is concerned, among other things, with the study of fair distributions of the outcomes of cooperation, has analysed various distribution rules based on egalitarian criteria. For instance, van den Brink
(2007) provides a comparison of the equal division value and the Shapley
value, and Casajus and H\"{u}ttner (2014) compare those two solutions with
the equal surplus division value (studied first in Driessen and Funaki, 1991). 

More recently, egalitarian solutions have been studied in the context of coalition-structured cooperative games. More specifically, Alonso-Meijide et al. (2020) extend the equal division and the equal surplus division values for cooperative games with a priori unions; in the case of the equal surplus division value Alonso-Meijide et al. (2020) propose three possible extensions. Hu and Li (2021) extend the equal surplus division value for cooperative games with level structures. Level structures were introduced in Owen (1977) and further studied in Winter (1989). A level structure is a collection of nested partitions of the set of players that conditions their negotiation. When instead of a collection of nested partitions we have a single partition, the level structure is called an a priori union structure. The Hu and Li's equal surplus division value for cooperative games with level structures turns out to be an extension of one of the three equal surplus division values for cooperative games with a priori unions introduced in Alonso-Meijide et al. (2020).

In this paper we extend to the case with levels the other two equal surplus division values for cooperative games with a priori unions introduced in Alonso-Meijide et al. (2020), as well as the equal division value. Moreover we provide new insights and axiomatic characterizations of all the values considered here, including the Hu and Li value, and illustrate their interest with a motivating example.

The structure of the paper is as follows. In Section~\ref{Sec:exa} we elaborate a particular cost sharing problem to motivate the informal presentation, and discussion, of various egalitarian values in the context of cooperative games with level structures. In Section~\ref{Sec:fev} we describe the model of cooperative games with level structures and define four egalitarian values. In Section~\ref{Sec:axic} we provide axiomatic characterizations of the previous values. Finally, Section~\ref{Sec:conre} is devoted to the concluding remarks.

\section{An example}\label{Sec:exa}

Housing legislation in most democratic states includes regulations on how to share the costs of improving common elements in homeowners' associations. In some cases, such regulations recommend the use of equal sharing criteria. For example, in the Netherlands, each of the owners of the dwellings involved in an improvement of the common elements of a building must share equally in the debts and costs involved, unless the internal community agreements provide for a different proportion of participation.

In this section we deal with an example that arises in the sharing of the ordinary maintenance costs of a parking area serving two residential buildings. The specific characteristics of the example are as follows. We consider a small residential complex consisting of two buildings that share an underground parking area. The first building is a two-storey building with one dwelling on each floor. The second building is also two-storey with one dwelling on the first floor and two dwellings on the second floor. From the parking area there is a lift to the first building and another lift to the second building. The owners of the dwellings in the residential complex have entrusted the maintenance of the common parking area (where each owner has one parking place) to a company which is responsible for the cleaning and security of the parking area, as well as for the maintenance of the lifts.

The company in charge of maintenance has a standard monthly fee for each community of owners which has two parts: a fixed part of 50 euros, and a variable part consisting of 50 euros for each lift, 4 euros times the highest floor the lift should reach (for each lift) and 10 euros for each parking place. Accordingly, the community in this example should pay the following monthly amount in euros (broken down as the sum of the fee components):
$$50+50\cdot 2+4\cdot 2\cdot 2+10\cdot 5=216.$$
Now the question is how to distribute this fee in an equitable way among the five owners involved. Below are several proposals.
\begin{itemize}
\item 
One proposal is to divide the total fee equally between the five owners. Accordingly, each of them would have to pay 43.2 euros.
\item 
Another proposal is to take into account that not all owners would have the same individual fee; by individual fee we mean the fee that each owner would have to pay if only he had access to the parking area. According to this proposal each owner should pay his individual fee plus the remainder of the total fee divided equally among the owners; by the remainder of the total fee we mean the difference between the total fee and the sum of the individual fees (notice that this remainder may be a negative amount). Table \ref{table1} shows the individual fee of each owner. Therefore, each owner must pay his individual fee plus $\frac{1}{5}(216-582)$, resulting in the following distribution vector: $(40.8,44.8,40.8,44.8,44.8)$.
\end{itemize}

\vspace*{0.5cm}

\begin{table}[htbp]
	\begin{center}
		\begin{tabular}{|c|c|}
			\hline
			Owner & Individual Fee \\ \hline
			1 & $50+50\cdot 1+4\cdot 1\cdot 1+10\cdot 1=114$  \\ \hline
			2 & $50+50\cdot 1+4\cdot 2\cdot 1+10\cdot 1=118$ \\ \hline
			3 & $50+50\cdot 1+4\cdot 1\cdot 1+10\cdot 1=114$ \\ \hline
			4 & $50+50\cdot 1+4\cdot 2\cdot 1+10\cdot 1=118$ \\ \hline
			5 & $50+50\cdot 1+4\cdot 2\cdot 1+10\cdot 1=118$ \\ \hline
		\end{tabular}%
	\end{center}
	\caption{Owners' individual fees}
	\label{table1}
\end{table}
\noindent
The two proposals above do not take into account that, in setting the monthly fee, it would be appropriate to take into account that the owners participate in the residential complex according to a nested structure. Firstly, the owners are grouped according to the lifts that serve them; secondly, they are grouped according to the floor on which their property is located. The nested structure of the owners is shown in Figure \ref{figure1}. Next we describe four other proposals for distributing the total fee that do take into account the nested structure of the owners.

\begin{figure}[htbp]
	\begin{center}
		\includegraphics[scale=0.75]{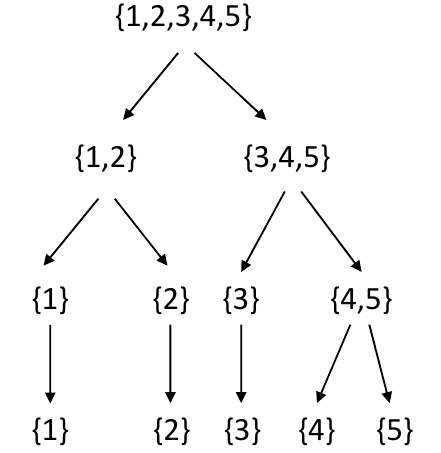}
	\end{center}
	\caption{The nested structure of the property owners}
	\label{figure1}
\end{figure}

\begin{itemize}
	\item 
	The first of these four proposals is an equal division at each level of the nested structure; i.e. first we divide the 216 euros equally between the two lifts, then we divide the amount allocated to each lift equally between the floors it serves; finally, we divide the amount allocated to each lift and floor equally between the corresponding owners. We say that the total fee is divided equally between the owners, but taking into account their nested structure. This proposal results in the following distribution vector: $(54,54,54,27,27)$.
\end{itemize}
\noindent
The following three proposals take into account the nested structure of the owners, but also the individual fees that the owners, floors and lifts would have to pay if being alone. The latter can be done in various ways, resulting in three different proposals.

\begin{itemize}
	\item 
	One proposal is to calculate the individual fee for each lift and share it equally among the owners it serves taking into account the nested structure. The remainder of the total fee is then divided equally among the owners taking again into account the nested structure. The individual fees for the lifts would be $128$ and $138$, respectively. The remainder of the total fee is $-50$. Then, for instance, the forth component of the corresponding distribution vector is 
	$$\frac{138}{2\cdot 2}-\frac{50}{2\cdot 2\cdot 2}=28.25$$
	and the complete distribution vector is $(51.5,51.5,56.5,28.25,28.25)$.
	\item 
	Another proposal is to calculate the individual fee for each owner, for each floor and for each lift and then allocate to each owner his individual fee, plus the remainder of the individual fee for his floor divided equally among the owners of his floor (taking into account the nested structure), plus the remainder of the individual fee for his lift divided equally among the owners of his lift (taking into account the nested structure), plus the remainder of the total fee divided equally among the owners (taking into account the nested structure). Table \ref{table2} shows the individual fees of floors and lifts.  Then, for instance, the forth component of the corresponding distribution vector is 
	$$118+\frac{128-2\cdot 118}{2}+\frac{138-114-128}{2\cdot 2}+\frac{216-128-138}{2\cdot 2\cdot 2}=31.75$$
	and the complete distribution vector is $(49.5,53.5,49.5,31.75,31.75)$.
	
	\vspace*{0.5cm}
	
	\begin{table}[htbp]
		\begin{center}
			\begin{tabular}{|c|c|}
				\hline
				& Individual Fee \\ \hline
				Floor 1, Lift 1 & $50+50\cdot 1+4\cdot 1\cdot 1+10\cdot 1=114$  \\ \hline
				Floor 2, Lift 1 & $50+50\cdot 1+4\cdot 2\cdot 1+10\cdot 1=118$ \\ \hline
				Floor 1, Lift 2 & $50+50\cdot 1+4\cdot 1\cdot 1+10\cdot 1=114$ \\ \hline
				Floor 2, Lift 2 & $50+50\cdot 1+4\cdot 2\cdot 1+10\cdot 2=128$ \\ \hline
				Lift 1 & $50+50\cdot 1+4\cdot 2\cdot 1+10\cdot 2=128$ \\ \hline
				Lift 2 & $50+50\cdot 1+4\cdot 2\cdot 1+10\cdot 3=138$ \\ \hline
			\end{tabular}%
		\end{center}
		\caption{Floor's and lifts' individual fees}
		\label{table2}
	\end{table}
\item 
The last proposal we put forward is to allocate to each owner his individual fee plus the remainder of the total fee divided equally among the owners taking into account the nested structure. Then, for instance, the forth component of the corresponding distribution vector is 
$$118+\frac{216-114-118-114-118-118}{2\cdot 2\cdot 2}=72.25$$
and the complete distribution vector is $(22.5,26.5,22.5,72.25,72.25)$.
\end{itemize}

Table \ref{table3} shows the six proposed fee allocations among the owners that we have calculated above. All of them are adaptations of the equal share criterion and all seem reasonable, at least theoretically. However, in practice, they are not all equally reasonable; for instance, in this example, the distribution of Proposal 6 seems difficult for the fourth and fifth owners to accept. 

 In order to decide in which cases it is more appropriate to use each of the six proposals it is necessary to analyse them mathematically. The first two proposals correspond to the equal division rule and the equal surplus division rule, which have been studied for example in van den Brink (2007) and in Casajus and H\"{u}ttner (2014), respectively. The other four will be analysed in the following sections of this article. Proposal 5 has been introduced and studied in Xu and Li (2021), but in this article we obtain a new axiomatic characterization for it.

\vspace*{0.5cm}

\begin{table}[htbp]
	\begin{center}
		\begin{tabular}{|c|c|c|c|c|c|}
			\hline
		 &1&2&3&4&5 \\ \hline
		Proposal 1 & 43.2 & 43.2 & 43.2 & 43.2 & 43.2 \\ \hline
		Proposal 2 & 40.8 & 44.8 & 40.8 & 44.8 & 44.8 \\ \hline
		Proposal 3 & 54 & 54 & 54 & 27 & 27 \\ \hline
		Proposal 4 & 51.5 & 51.5 & 56.5 & 28.25 & 28.25  \\ \hline
		Proposal 5 & 49.5& 53.5&49.5 & 31.75& 31.75 \\ \hline
		Proposal 6 &22.5 &26.5 & 22.5& 72.25& 72.25 \\ \hline
		\end{tabular}%
	\end{center}
	\caption{The six proposals}
	\label{table3}
\end{table}

\section{Four egalitarian values for cooperative games with level structures}\label{Sec:fev}

Transferable utility games (abbreviated, TU-games) are probably the most studied model of cooperative 
game theory. A TU-game is defined by a finite set of players $N=\{1,2,\ldots,n\}$ and a real-valued 
function  $v:2^N\rightarrow\mathbb{R}$, called a characteristic function, where $v(S)$ is the worth 
of $S \subseteq N$, that is, the benefits (or costs) that coalition $S$ is able to generate. 
By convention, $v(\emptyset)=0$. 
A TU-game with a priori unions is a triplet $(N, v, \mathcal{C})$, where 
$\mathcal{C}=\{C_1,C_2,\ldots ,C_m\}$ is a partition of the player set $N$ into unions.

A level structure  $\mathcal{L}$ is a sequence of nested partitions of the set of players $N$,  
$\mathcal{L}=\{\mathcal{C}_0,\mathcal{C}_1,\ldots,\mathcal{C}_{k+1}\}$. Level zero comprises each 
player $i\in N$ as a union, $\mathcal{C}_0 = \{\{i\} : i\in N \}$. At each level $l\geq 1$, 
the partition is obtained by aggregation of unions from level $l-1$; if $C\in \mathcal{C}_l$ 
then $C=\bigcup_{C'\in \hat{\mathcal{C}}_{l-1}}C'$, where 
$\hat{\mathcal{C}}_{l-1}\subseteq\mathcal{C}_{l-1}$. 
Finally, $\mathcal{C}_{k+1}$ is the trivial partition with the grand coalition $N$ as the 
sole union, $\mathcal{C}_{k+1}=\{ N \}$. A level game is a triplet $(N,v,\mathcal{L})$, 
where $(N,v)$ is a TU-game and $\mathcal{L}$ is a level structure. A TU-game with a priori 
unions can be viewed as a particular case of a level game with $k=1$.

For every player $i\in N$ and level $l\in\{0,1,\ldots,k+1\}$, we denote by $C_l(i)$ the union 
in $\mathcal{C}_l\in\mathcal{L}$ containing $i$, that is, $i \in C_l(i) \in \mathcal{C}_l$. 
Also, for every level $l\in\{1,2,\ldots,k+1\}$ and union $C\in\mathcal{C}_l\in\mathcal{L}$, 
we denote by $\lfloor C \rfloor$ the set of all unions in $\mathcal{C}_{l-1}$ 
contained in $C$. We can refer to these unions as the direct subordinates of $C$.
$$\lfloor C\rfloor=\{C'\in\mathcal{C}_{l-1}\in\mathcal{L} : C'\subseteq C\}.$$

By a value, we mean a map $g$ that assigns to every level game $(N,v,\mathcal{L})$ 
a vector $g (N,v,\mathcal{L})  \in \mathbb{R}^N$ with components $g_i (N,v,\mathcal{L})$ 
for all $i \in N$. Each component represents the player's payoff according to $g$. 
Alonso-Meijide et al. (2020)  define and study values for games with a priori unions 
using an egalitarian approach. Next, we propose four values that extend the previous 
ones to the context of level games. The third value has already been introduced in Hu and Li (2021).

The first proposal divides the worth of the grand coalition, $v(N)$, equally among 
the players at each level of the level structure, that is, for a player $i\in N$, 
the value divides $v(N) = v(C_{k+1}(i))$ by the cardinals of the direct subordinates 
of all unions containing $i$. 
We call this an equitable allocation among the players, but taking into account the 
level structure.

\begin{definition}
	The LED-value is defined for every $(N,v,\mathcal{L})$ and $i\in N$ by
	$$LED_{i}=\frac{v(C_{k+1}(i))}{\Pi_{l=1}^{k+1}|\lfloor C_l(i)\rfloor|}.$$
\end{definition}

The second proposal is to calculate the worth of each union at level $k$ 
(the last non-trivial partition) and share it equally among its players, 
but taking into account the level structure.
The remainder of the total worth is then divided equally among the players 
taking again into account the level structure.

\begin{definition}
	The LESD$^{1}$-value is defined for every $(N,v,\mathcal{L})$ and $i\in N$ by
	$$LESD^1_{i}=\frac{v(C_k(i))}{\Pi_{l=1}^{k}|\lfloor C_l(i)\rfloor|}+\frac{v(C_{k+1}(i))-\sum_{C\in \lfloor C_{k+1}(i)\rfloor}v(C)}{\Pi_{l=1}^{k+1}|\lfloor C_l(i)\rfloor|}.$$
\end{definition}

The third proposal is to allocate to each player its individual worth, 
plus the remainder of the worth of the union to which it belongs on level 1 
divided equally among the members of the union, but taking into account the 
level structure, and so on.

\begin{definition}
	The LESD$^2$-value is defined for every $(N,v,\mathcal{L})$ and $i\in N$ by
	$$LESD^2_{i}=v(i)+\sum_{l=1}^{k+1}\frac{v(C_{l}(i))-\sum_{C\in \lfloor C_{l}(i)\rfloor}v(C)}{\Pi_{l'=1}^{l}|\lfloor C_{l^{'}}(i)\rfloor|}.$$
\end{definition}

The fourth and last proposal is to allocate to each player its individual worth,
plus the remainder of the worth of the grand coalition divided equally among 
the players, but taking into account the level structure.

\begin{definition}
	The LESD$^3$-value is defined for every $(N,v,\mathcal{L})$ and $i\in N$ by
	$$LESD^3_{i}=v(i)+\frac{v(C_{k+1}(i))-\sum_{j\in N}v(j)}{\Pi_{l=1}^{k+1}|\lfloor C_l(i)\rfloor|}.$$
\end{definition}

To conclude this section, we introduce a few new concepts that we will need in the next section, in which we provide axiomatic characterizations of the four values.
Given a level game $(N,v,\mathcal{L})$ and level $l\in\{0,1,\ldots,k\}$, we define 
$v^l$, a new characteristic function on $\mathcal{C}_l$, as
$$v^l(S)=v\left(\bigcup_{C\in S}C\right)\text{ for all }S\subseteq\mathcal{C}_l.$$
The $l$-th quotient game $(\mathcal{C}_l,v^l)$ is  induced from $(N,v,\mathcal{L})$ 
by treating unions of $\mathcal{C}_l\in\mathcal{L}$ as players.

Given a level structure $\mathcal{L}=\{\mathcal{C}_0,\mathcal{C}_1,\ldots,\mathcal{C}_{k+1}\}$ 
and level $l\in\{0,1,\ldots,k\}$, we define the $l$-th truncation of 
$\mathcal{L}$, denoted by $\mathcal{L}^l$, as a new level structure in which the 
set of players is $\mathcal{C}_l\in\mathcal{L}$ and 
$\mathcal{L}^l=\{\mathcal{C}_0^l,\mathcal{C}_1^l,\ldots,\mathcal{C}_{k-l+1}^l\}$, with
\begin{itemize}
	\item $\mathcal{C}_0^l=\{ \{S\} : S\in \mathcal{C}_l \}$.
	\item $\mathcal{C}_j^l=\{\{T\in \mathcal{C}_0^l : T\subseteq Q\} : Q\in \mathcal{C}_{l+j}\}$, 
	$j=1,2,\ldots,k-l+1$.
\end{itemize}
We call $(\mathcal{C}_l, v^l, \mathcal{L}^l)$ the $l$-th truncated game.

Given a TU-game $(N,v)$ and players $i,j\in N$, we say that $i,j$ are indistinguishable in $(N,v)$ if $v(S\cup i)=v(S\cup j)$ for all $S\subseteq N\backslash\{i,j\}$. We say that $i$ is a nullifying player in $(N,v)$ if $v(S\cup i)=0$ for all $S\subseteq N$. We say that $i$ is a dummifying player in $(N,v)$ if $v(S\cup i)=\sum_{j \in S \cup i} v(j)$ for all $S\subseteq N$. The above definitions extend to level games in a natural way; thus, for example, we say that $i,j$ are indistinguishable in $(N,v,\mathcal{L})$ if they are indistinguishable in $(N,v)$. Finally, we denote the restriction of $(N,v)$ on $S\subset N$ as $(S,v)$.

\section{Axiomatic characterizations}\label{Sec:axic}

In this section we provide axiomatic characterizations of the four values defined in Section \ref{Sec:fev}. 

We show below a set of properties that characterizes the first proposal. The first two properties are standard in the literature. The property of symmetry among unions on each level says that two indistinguishable coalitions at a particular level, both included in the same coalition at the next level, receive the same. The nullifying player property states that a nullifying player receives zero. 
\bigskip

\noindent \textbf{Efficiency.} A value $g$ for level games satisfies efficiency if, for all $(N,v,\mathcal{L})$, it holds that
$$\sum_{i\in N}g_i(N,v,\mathcal{L})=v(N).$$

\noindent \textbf{Additivity.} A value $g$ for level games satisfies additivity if, for all  $(N,v,\mathcal{L})$ and $(N,w,\mathcal{L})$, it holds that
$$g(N,v,\mathcal{L})+g(N,w,\mathcal{L})=g(N,v+w,\mathcal{L}).$$

\noindent 
\textbf{Symmetry among unions on each level.} A value $g$ for level games satisfies symmetry among unions on each level if, for all $(N,v,\mathcal{L})$, with $\mathcal{L}=\{\mathcal{C}_0,\mathcal{C}_1,\ldots,\mathcal{C}_{k+1}\}$, $l\in\{0,1,\ldots,k\}$ and $ C, C'\in \mathcal{C}_l$, with $C\cup C'\subseteq C''\in \mathcal{C}_{l+1}$, indistinguishable in $(\mathcal{C}_l,v^l,\mathcal{L}^l)$, it holds that
$$\sum_{i\in C}g_i(N,v,\mathcal{L})=\sum_{j\in C'}g_j(N,v,\mathcal{L}).$$

\noindent \textbf{Nullifying player property.}  A value $g$ for level games satisfies the nullifying player property if, for all $(N,v,\mathcal{L})$ and all $i\in N$ nullifying player in $(N,v,\mathcal{L})$, it holds that
$$g_i(N,v,\mathcal{L})=0.$$

\begin{theo}
	\label{thed}
	The LED-value is the unique value for level games that satisfies efficiency, additivity, symmetry among unions on each level and nullifying player property.
\end{theo}
\begin{proof}
	It is immediate to check that the LED-value satisfies efficiency, additivity, symmetry among unions on each level and nullifying player property. To prove the unicity, consider a value $g$ for level games satisfying efficiency, additivity,  symmetry among unions on each level and nullifying player property. Fix $N$ and define, for all $\alpha \in \mathbb{R}$ and all non-empty $T\subseteq N$, the TU-game $(N,e_{T}^{\alpha
	}) $ given by $e_{T}^{\alpha }(S)=\alpha $ if $S=T$ and $e_{T}^{\alpha
	}(S)=0 $ if $S\neq T$. If $T=N$, since $g$ satisfies efficiency and symmetry among unions on each level, it is 	clear that 
	$$g_{i}(N,e_{N}^{\alpha },\mathcal{L})=\frac{\alpha }{\Pi_{l=1}^{k+1}|\lfloor C_l(i)\rfloor|}$$ 
	for any  $i\in N$ and $\mathcal{L}=\{\mathcal{C}_0,\mathcal{C}_1,\ldots,\mathcal{C}_{k+1}\}$, because all unions in $\mathcal{C}_l$ are indistinguishable in $(\mathcal{C}_l,(e_{N}^{\alpha})^l,\mathcal{L}^l)$, with  $l\in\{0,1,\ldots,k\}$. If $T\subset N$, notice
	that all players in $N\setminus T$ are nullifying players in $(N,e_{T}^{\alpha
	}) $ and then, since $g$ satisfies efficiency and nullifying player property, 
	\begin{equation*}
		\sum_{i\in T}g_{i}(N,e_{T}^{\alpha },\mathcal{L})=\sum_{i\in N}g_{i}(N,e_{T}^{\alpha
		},\mathcal{L})=e_{T}^{\alpha }(N)=0.
	\end{equation*}%
	Then, since $g$ satisfies symmetry among unions on each level it is not difficult to check that 
	$g_i(N,e_{T}^{\alpha },\mathcal{L})=0$ for all $i\in N$. Finally, the additivity of $g$ and the fact that $v=\sum_{T\subseteq N}e_{T}^{v(T)}$ imply that 
	\begin{equation*}
		g_{i}(N,v,\mathcal{L})=\sum_{T\subseteq
			N}g_{i}(N,e_{T}^{v(T)},\mathcal{L})=g_{i}(N,e_{N}^{v(N)},\mathcal{L})=\frac{v(N)}{\Pi_{l=1}^{k+1}|\lfloor C_l(i)\rfloor|}
	\end{equation*}
	and thus
	\begin{equation*}
		g(N,v,\mathcal{L})=LED(N,v,\mathcal{L}).
	\end{equation*}
\end{proof}\bigskip 

The property of dummifying level/nullifying player property says that if a player of a dummifying union in the last non trivial partition is a nullifying player in the game restricted to this union, receives zero. This property replaces the nullifying player property in the characterization of the first proposal.

\bigskip 

\noindent 
\textbf{Dummifying level/nullifying player property.} A value $g$ for level games satisfies the
dummifying level/nullifying player property if, for all $(N,v,\mathcal{L})$, with $\mathcal{L}=\{\mathcal{C}_0,\mathcal{C}_1,\ldots,\mathcal{C}_{k+1}\}$, and all $i\in N$ nullifying player in $(C_k(i),v)$ such that $C_k(i)$ is a dummifying player in $(\mathcal{C}_{k},v^{k})$, then it holds that $$g_{i}(N,v,\mathcal{L})=0.$$

\begin{theo}
	The LESD$^1$-value is the unique value  for level games that satisfies efficiency, additivity, symmetry among unions on each level and dummifying level/nullifying player property.
\end{theo}
\begin{proof} It is immediate to check that the LESD$^1$-value satisfies efficiency, additivity, symmetry among unions on each level and dummifying level/nullifying player property. To prove the unicity, consider a value $g$ for level games satisfying efficiency, additivity, symmetry among unions on each level and dummifying level/nullifying player property. Take $(N,v,\mathcal{L})$ with $\mathcal{L}=\{\mathcal{C}_0,\mathcal{C}_1,\ldots,\mathcal{C}_{k+1}\}$ and $\mathcal{C}_{k}=\{C_1,C_2,\ldots,C_m\}$. Now, define the TU-game $(N,v^{*})$ given by
	
	\begin{equation*}
		v^*(S) = \sum_{C\in \mathcal{C}_{k}:C\subseteq S}v(C) = \sum_{j=1}^{m}v^{C_{j}}(S)
	\end{equation*}
	for all $S\subseteq N$, where $v^{C_j}(S) = v(C_j)$ if $C_{j}\subseteq S$ 
	and $v^{C_j}(S) = 0$ otherwise. 
	Take $C_r\in \mathcal{C}_{k}$. Since $g$ satisfies efficiency, then
	\begin{equation*}
		\sum_{i\in N}g_i(N,v^{C_r},\mathcal{L})=v^{C_r}(N)=v(C_r).
	\end{equation*}
	All unions $C_j\in \mathcal{C}_{k}$ are dummifying players in $(\mathcal{C}_{k},(v^{C_r})^k)$ and all players $i\in C_j$, with $j\neq r$, are nullifying players in $(C_k(i),v^{C_r})$. By dummifying level/nullifying player property, $g_i(N,v^{C_r},\mathcal{L})=0$ for all $i\notin C_r$. 
	For each level $l\in\{0,1,\ldots,k-1\}$, unions $C, C'\in \mathcal{C}_l$, with $C\cup C'\subseteq C''\in \mathcal{C}_{l+1}$ and $C\cup C'\subseteq C_r$, are indistinguishable in $(\mathcal{C}_l,(v^{C_r})^l,\mathcal{L}^l)$, then symmetry among unions on each level implies that, for all $i\in C_r$, 
	$$g_i(N,v^{C_r},\mathcal{L}) = \frac{v(C_r)}{\Pi_{l=1}^{k}|\lfloor C_l(i)\rfloor|}=\frac{v(C_k(i))}{\Pi_{l=1}^{k}|\lfloor C_l(i)\rfloor|}.$$ 
	Using the additivity of $g$, for all $i\in N$, 
	\begin{equation}  \label{eq200}
		g_i(N,v^*,\mathcal{L})=\frac{v(C_k(i))}{\Pi_{l=1}^{k}|\lfloor C_l(i)\rfloor|}.
	\end{equation}
	Define now the characteristic functions $v^{**} = v - v^*$ and, for all $\alpha \in \mathbb{R}$ and all non-empty $T\subseteq N$, 
	$e_{T}^{\alpha}$ given by 
	$e_{T}^{\alpha }(S)=\alpha$ if $S=T$ and $e_{T}^{\alpha}(S)=0$ if $S\neq T$. 
	It is clear that $v^{**}=\sum_{T\subseteq N}e^{v^{**}(T)}_T$. If $T=N$,
	since $g$ satisfies efficiency and symmetry among unions on each level, it is
	clear that, for all $i\in N$, $g_i(N,e^{v^{**}(N)}_N,\mathcal{L})$ is given by
	\begin{equation*}
		\frac{{v^{**}(N)}}{\Pi_{l=1}^{k+1}|\lfloor C_l(i)\rfloor|} 
		= \frac{v(N) - \sum_{C\in  \mathcal{C}_{k}}v(C)}{\Pi_{l=1}^{k+1}|\lfloor C_l(i)\rfloor|}=\frac{v(C_{k+1}(i))-\sum_{C\in \lfloor C_{k+1}(i)\rfloor}v(C)}{\Pi_{l=1}^{k+1}|\lfloor C_l(i)\rfloor|}. 
	\end{equation*}
	If $T\subset N$, consider two cases:
	
	\begin{itemize}
		\item Take $T=\cup_{C\in R}C$, with $\emptyset\subset R\subset \mathcal{C}_{k}$. 
		For all $C_j\in \mathcal{C}_{k}$, if $T\neq C_j$ then $e_T^{v^{**}(T)}(C_j)=0$ and 
		if $T=C_j$ then $e_T^{v^{**}(T)}(C_j)=v^{**}(C_j)=0$. Hence, it is easy to 
		see that all unions in $\mathcal{C}_{k}\setminus R$ are dummifying players in 
		$(\mathcal{C}_{k},(e_T^{v^{**}(T)})^k)$. Also, since all players $i\in N\setminus T$ are 
		nullifying players in $(C_k(i),e^{v^{**}(T)}_T)$, dummifying level/nullifying player property implies that $g_i(N,e^{v^{**}(T)}_T,\mathcal{L})=0$ for all $i\in N\setminus T$. Notice that since all unions in $R$ are indistinguishable in  $(\mathcal{C}_{k},(e_T^{v^{**}(T)})^k,\mathcal{L}^k)$, then by symmetry among unions on each level, 
		$\sum_{i\in C}g_i(N,e^{v^{**}(T)}_T,\mathcal{L})
		=\sum_{i\in C'}g_i(N,e^{v^{**}(T)}_T,\mathcal{L})$ for all $C,C'\in R$; 
		notice also that, since
		\begin{equation*}
			\sum_{i\in T} g_i(N,e^{v^{**}(T)}_T,\mathcal{L}) = \sum_{i\in N} g_i(N,e^{v^{**}(T)}_T,\mathcal{L}) 
			=  e^{v^{**}(T)}_T(N) = 0,
		\end{equation*}
		then $\sum_{i\in C}g_i(N,e^{v^{**}(T)}_T,\mathcal{L})=0$ for all $C\in R$. To conclude, symmetry among unions on each level implies that $g_i(N,e^{v^{**}(T)}_T,\mathcal{L})=0$ for all 
		$i\in C\in R$, and therefore for all $i\in T$.
		
		\item For any other $T\subset N$ that is not in the previous case, 
		the  game $(\mathcal{C}_{k},(e_T^{v^{**}(T)})^k)$ satisfies that 
		$(e_T^{v^{**}(T)})^k(R)=0$ for all $R\subseteq \mathcal{C}_{k}$ and, thus, all unions in 
		$\mathcal{C}_{k}$ are indistinguishable and dummifying players in $(\mathcal{C}_{k},(e_T^{v^{**}(T)})^k)$. 
		If $i\in N\setminus T$, then $i$ is a nullifying player in $(C_k(i),e^{v^{**}(T)}_T)$ 
		and dummifying level/nullifying player property implies that $g_i(N,e^{v^{**}(T)}_T,\mathcal{L})=0$. 
		Analogously as in the previous case, symmetry among unions on each level implies that 
		$g_i(N,e^{v^{**}(T)}_T,\mathcal{L})=0$ for all $i\in T$.
	\end{itemize}
	Now additivity implies that, for all $i\in N$,  
	\begin{equation}  \label{eq201}
		g_i(N,v^{**},\mathcal{L})=\sum_{T\subseteq N}g_i(N,e_T^{v^{**}(T)},\mathcal{L})=\frac{{v^{**}(N)}}{\Pi_{l=1}^{k+1}|\lfloor C_l(i)\rfloor|}.
	\end{equation}
	Finally, from (\ref{eq200}), (\ref{eq201}), additivity and $v = v^{*} + v^{**}$ it is clear that 
	\begin{equation*}
		g(N,v,\mathcal{L})=LESD^1(N,v,\mathcal{L}).
	\end{equation*}\end{proof}
\bigskip

To characterize the third proposal, we use the same set of properties included in the previous theorems to characterize the first and second proposal, with a unique modification. The property of dummifying unions for a player replaces the nullifying player property of the first proposal or the dummifying level/nullifying player property of the second. The dummifying unions for a player property states that if the union containing a particular player is a dummifying union at every level, then this player receives his/her individual worth.

\bigskip \noindent \textbf{Dummifying unions for a player property.} A value 
$g$ for level games satisfies the dummifying unions for a player property if, for all $(N,v,\mathcal{L})$ with $\mathcal{L}=\{\mathcal{C}_0,\mathcal{C}_1,\ldots,\mathcal{C}_{k+1}\}$, and all $i\in N$ such that  $C_l(i)\in \mathcal{C}_l$ is a dummifying player in $(\mathcal{C}_{l},v^{l})$ for all $l\in\{0,1,\ldots,k\}$, then it holds that $$g_{i}(N,v,\mathcal{L})=v(i).$$

\begin{theo}
	The LESD$^2$-value is the unique value for level games that satisfies efficiency, additivity, symmetry among unions on each level and dummifying unions for a player property.
\end{theo}
\begin{proof} It is immediate to check that the LESD$^2$-value satisfies efficiency, additivity, symmetry among unions on each level and dummifying unions for a player property. 
	To prove the unicity, consider a value $g$ for level games that satisfies efficiency, additivity, symmetry among unions on each level and dummifying unions for a player property. Take $(N,v,\mathcal{L})$ with $\mathcal{L}=\{\mathcal{C}_{0},\mathcal{C}_{1},\ldots,\mathcal{C}_{k+1}\}$ and define $v^a$, for all  $S\subseteq N$, by
	\begin{equation*}
		v^a(S)=\sum_{i\in S}v(i).
	\end{equation*}
	Define $\tilde{v}$, for all  $S\subseteq N$, by
	\begin{equation*}
		\tilde{v}(S)=(v-v^a)(S)-\sum_{C\in \mathcal{C}_{k}:C\subseteq S} (v-v^a)(C).
	\end{equation*}
	Now define $v^*_l$, for all $l\in\{1,2,\ldots,k\}$ with $\mathcal{C}_{l}=\{C_{l,1},C_{l,2},\ldots,C_{l,|\mathcal{C}_{l}|}\}$, by
	\begin{equation*}
		v^*_l(S) = \sum_{C\in \mathcal{C}_{l}:C\subseteq S} \left(v(C)-\sum_{C'\in \mathcal{C}_{l-1}:C'\subseteq C} v(C')\right) = \sum_{j=1}^{|\mathcal{C}_{l}|} v^{C_{l,j}}_l(S)
	\end{equation*}
	for all $S\subseteq N$, where $v^{C_{l,j}}_l(S)=v(C_{l,j})-\sum_{C'\in \mathcal{C}_{l-1}:C'\subseteq C_{l,j}} v(C')$ if $C_{l,j}\subseteq S$ and $v^{C_{l,j}}_l(S)=0$ otherwise. 
	For all $l\in\{0,1,\ldots,k\}$,  $C_l(i)\in \mathcal{C}_l$ is a dummifying player in $(\mathcal{C}_{l},(v^{a})^l)$; hence, dummifying unions for a player property implies that, for all $i\in N$,
	\begin{equation}  \label{eq199}
		g_i(N,v^{a},\mathcal{L})=v^a(i)=v(i).
	\end{equation}
	Take now $l\in\{1,2,\ldots,k\}$ and $C_r\in \mathcal{C}_{l}$. Since $g$ satisfies efficiency, then
	\begin{equation*}
		\sum_{i\in N}g_i(N,v^{C_r}_l,\mathcal{L})=v^{C_r}_l(N)=v(C_{r})-\sum_{C'\in \mathcal{C}_{l-1}:C'\subseteq C_{r}} v(C').
	\end{equation*}
	For all $l'\in\{0,1,\ldots,l\}$, all unions $C\in \mathcal{C}_{l'}$ such that $C\nsubseteq C_r$ are dummifying players in $(\mathcal{C}_{l'},(v^{C_r}_l)^{l'})$. By dummifying unions for a player property, $g_i(N,v^{C_r}_l,\mathcal{L})=v^{C_r}_l(i)=0$ for all $i\notin C_r$. And
	since all players in $C_r$ are indistinguishable in $(N,v^{C_r}_l,\mathcal{L})$, then 
	symmetry among unions on each level implies that, for all $i\in C_r$, 
	\begin{equation*}
		g_i(N,v^{C_r}_l,\mathcal{L})=\frac{v^{C_r}_l(C_r)}{\Pi_{l'=1}^{l+1}|\lfloor C_{l'}(i)\rfloor|}.
	\end{equation*}
	Using additivity, for all $i\in N$, 
	\begin{equation}  \label{eq198}
		g_i(N,v^*_l,\mathcal{L})=\frac{v^{C_l(i)}_l(C_l(i))}{\Pi_{l'=1}^{l+1}|\lfloor C_{l'}(i)\rfloor|}.
	\end{equation}
	Define, for all $\alpha \in \mathbb{R}$ and all non-empty $T\subseteq N$, the TU-game $(N,e_{T}^{\alpha}) $ given by $e_{T}^{\alpha }(S)=\alpha $ if $S=T$ and $e_{T}^{\alpha}(S)=0 $ if $S\neq T$.
	Take now into account that $\tilde{v}=\sum_{T\subseteq N}e^{\tilde{v}(T)}_T$. If $T=N$, since $g$ satisfies efficiency and symmetry among unions on each level, it is clear that for all $i\in N$, 
	\begin{equation*}
		g_i(N,e^{\tilde{v}(N)}_N,\mathcal{L})=\frac{{\tilde{v}(N)}}{\Pi_{l=1}^{k+1}|\lfloor C_l(i)\rfloor|}. 
	\end{equation*}
	If $T\subset N$, consider two cases:
	
	\begin{itemize}
		\item Take $T=\cup_{C\in R}C$, with $\emptyset\subset R\subset \mathcal{C}_{k}$. 
		For all $C_j\in \mathcal{C}_{k}$, if $T\neq C_j$ then $e_T^{\tilde{v}(T)}(C_j)=0$ and 
		if $T=C_j$ then $e_T^{\tilde{v}(T)}(C_j)=\tilde{v}(C_j)=0$. Also, if $Q\subseteq \mathcal{C}_{k}$ with $(\mathcal{C}_{k}\setminus R)\cap Q \neq\emptyset$, then $e_T^{\tilde{v}(T)}(\cup_{C\in Q}C)=0$.
		For all $i\in N\setminus T$ and $l\in\{0,1,\ldots,k\}$,  $C_l(i)\in \mathcal{C}_l$ is a dummifying player in $(\mathcal{C}_{l},(e^{\tilde{v}(T)}_T)^l,\mathcal{L}^l)$, then dummifying unions for a player property implies that $g_{i}(N,e^{\tilde{v}(T)}_T,\mathcal{L})=e^{\tilde{v}(T)}_T(i)=0$. Notice that all unions in $R$ are indistinguishable in $(\mathcal{C}_{k},(e^{\tilde{v}(T)}_T)^k,\mathcal{L}^k)$, then by symmetry among unions on each level, 
		$\sum_{i\in C}g_i(N,e^{\tilde{v}(T)}_T,\mathcal{L})=\sum_{i\in C'}g_i(N,e^{\tilde{v}(T)}_T,\mathcal{L})$ 
		for all $C,C'\in R$; and notice also that, since
		\begin{equation*}
			\sum_{i\in T} g_i(N,e^{\tilde{v}(T)}_T,\mathcal{L}) = 
			\sum_{i\in N} g_i(N,e^{\tilde{v}(T)}_T,\mathcal{L}) = e^{\tilde{v}(T)}_T(N) = 0,
		\end{equation*}
		then $\sum_{i\in C}g_i(N,e^{\tilde{v}(T)}_T,\mathcal{L})=0$ for all $C\in R$. Hence, symmetry among unions on each level implies 
		that $g_i(N,e^{\tilde{v}(T)}_T,\mathcal{L})=0$ for all $i\in T$.
		
		\item For any other $T\subset N$ that is not in the previous case,	notice that $e_T^{\tilde{v}(T)}(\cup_{C\in Q}C)=0$ for all $Q\subseteq \mathcal{C}_{k}$. If $i\in N\setminus T$, union $C_l(i)\in \mathcal{C}_l$ is a dummifying player in $(\mathcal{C}_{l},(e^{\tilde{v}(T)}_T)^l,\mathcal{L}^l)$ for all $l\in\{0,1,\ldots,k\}$, then by dummifying unions for a player property, $g_{i}(N,e^{\tilde{v}(T)}_T,\mathcal{L})=e^{\tilde{v}(T)}_T(i)=0$.
		Analogously as in the previous case, symmetry among unions on each level implies that 
		$g_i(N,e^{\tilde{v}(T)}_T,\mathcal{L})=0$ for all $i\in T$. 
	\end{itemize}
	Now additivity implies that, for all $i\in N$, 
	\begin{equation}  \label{eq197}
		g_i(N,\tilde{v},\mathcal{L})=\sum_{T\subseteq N}g_i(N,e_T^{\tilde{v}(T)},\mathcal{L})
		=\frac{{\tilde{v}(N)}}{\Pi_{l=1}^{k+1}|\lfloor C_l(i)\rfloor|}.
	\end{equation}
	Finally, from (\ref{eq199}), (\ref{eq198}), (\ref{eq197}), additivity and 
	$v = v^{a} + \tilde{v} + \sum_{l=1}^{k}v^{*}_l$ it is clear that 
	\begin{equation*}
		g(N,v,\mathcal{L})=LESD^2(N,v,\mathcal{L}).
	\end{equation*}
\end{proof}\bigskip

Finally, we present a set of properties that characterize the last proposal. Efficiency and additivity are common to the previous characterizations. The two new properties are: weak symmetry among unions on each level and dummifying player property. The dummifying player property states that a dummifying player receives his/her individual worth. Similar to the nullifying player of the first proposal, this property only takes into account the characteristic function of the game, and does not depend on the level structure. The weak symmetry among unions property on each level is a weaker version of the symmetry among unions property on each level (used to characterize the first proposal) because only indistinguishable coalitions formed by players with individual worth equal to zero receive the same. It is clear that if a value satisfies the property of symmetry among unions on each level, it also satisfies the weak property of symmetry among unions on each level. Then, there is not exist a value that satisfies efficiency, additivity, symmetry among unions on each level and dummifying player property. 

\bigskip \noindent \textbf{Dummifying player property.} A value $g$ for level games satisfies the dummifying player property if, for all $(N,v,\mathcal{L})$ and all $i\in N$ dummifying player in $(N,v,\mathcal{L})$, it holds that $$g_{i}( N,v,\mathcal{L}) =v(i).$$

\bigskip\noindent \textbf{Weak symmetry among unions on each level.} A value $g$ for level games satisfies weak symmetry among unions on each level  if, for all $(N,v,\mathcal{L})$, with $\mathcal{L}=\{\mathcal{C}_0,\mathcal{C}_1,\ldots,\mathcal{C}_{k+1}\}$ and $v(i)=0$ for all $i\in N$, $l\in\{0,1,\ldots,k\}$ and $C,C'\in \mathcal{C}_l$, with $C\cup C' \subseteq C''\in \mathcal{C}_{l+1}$, indistinguishable in $(\mathcal{C}_l,v^l,\mathcal{L}^l)$, it holds that
$$\sum_{i\in C}g_i(N,v,\mathcal{L})=\sum_{j\in C'}g_j(N,v,\mathcal{L}).$$

\begin{theo}
	The LESD$^3$-value is the unique value for level games that satisfies efficiency, additivity, weak symmetry among unions on each level and dummifying player property.
\end{theo}
\begin{proof} It is immediate to check that the LESD$^3$-value satisfies efficiency, additivity, weak symmetry among unions on each level and dummifying player property. To prove the unicity, consider a value $g$ for level games satisfying efficiency, additivity, weak symmetry among unions on each level and dummifying player property. Take $(N,v,\mathcal{L})$ and define $v^a$, for all  $S\subseteq N$, by
	\begin{equation*}
	v^a(S)=\sum_{i\in S}v(i).
	\end{equation*}
	Define $v^0=v-v^a$.	Additivity implies that 
	\begin{equation}  \label{eq899}
		g(N,v,\mathcal{L})=g(N,v^{a},\mathcal{L})+g(N,v^0,\mathcal{L}).
	\end{equation}
	Since $i$ is a dummifying player in $(N,v^a)$, for all $i\in N$, then dummifying player property implies that 
	\begin{equation}  \label{eq799}
		g_i(N,v^{a},\mathcal{L})=v^a(i)=v(i).
	\end{equation}
	Now, using for $(N,v^0, \mathcal{L})$ analogous arguments as those used in the proof of Theorem \ref{thed}, it is easy to see that efficiency, additivity, weak symmetry among unions on each level and dummifying player property imply that
	\begin{equation}  \label{eq898}
		g(N,v^0,\mathcal{L})=LED(N,v^0,\mathcal{L}).
	\end{equation}
	Finally, from (\ref{eq899}), (\ref{eq799}) and (\ref{eq898}) it is clear that 
	\begin{equation*}
		g(N,v,\mathcal{L})=LESD^3(N,v,\mathcal{L}).
\end{equation*}\end{proof}

\section{Concluding remarks}\label{Sec:conre}
In this article we have extended the equal division value and the equal surplus division value to level games. In the latter case, we have relied on three extensions of such a value for cooperative games with a priori unions described in Alonso-Meijide et al. (2020). In Gon\c{c}alves-Dosantos and Alonso-Meijide (2021)  two new variants of the equal surplus value for cooperative games with a priori unions are proposed. A topic for future work is the extension of these two variants to games with levels.

One of the advantages of the values based on the equal sharing criterion is that they are generally calculable in moderate times even for games with many players, because they do not make use of the full characteristic function of the game. In any case, in order for managers to be able to make use of the values considered here, it would be convenient to make available to the community a computer tool developed in free software to calculate them. In the near future it would be interesting to generate such a tool.

\section*{References}

\noindent Alonso-Meijide JM, Costa J, Garc\'{\i}a-Jurado I, Gon\c{c}alves-Dosantos JC (2020). On egalitarian values for cooperative games with a priori unions. TOP 28, 672-688.\newline
\noindent Casajus A, H\"{u}ttner F (2014). Null, nullifying, or dummifying
players: The difference between the Shapley value, the equal division value,
and the equal surplus division value. Economics Letters 122, 167-169.\newline
\noindent Driessen TSH, Funaki Y (1991). Coincidence of and collinearity 
between game theoretic solutions. OR Spectrum 13, 15-30.\newline
\noindent
Gon\c{c}alves-Dosantos JC, Alonso-Meijide JM (2021). New results on egalitarian values for games with a priori unions. To appear in Optimization, DOI: 10.1080/02331934.2021.1995731\newline
\noindent Hu XF, Li DF (2021). The equal surplus division value for cooperative games with a level structure. Group Decision and Negotiation 30, 1315-1341.\newline
\noindent Owen G (1977) Values of games with a priori unions. In: Mathematical
Economics and Game Theory (R Henn, O Moeschlin, eds.), Springer, 76-88.\newline
\noindent Selten R (1972). Equal share analysis of characteristic function
experiments. In: Contributions to Experimental Economics III. (Sauermann H,
ed.), Mohr Siebeck, 130-165.\newline
\noindent van den Brink R (2007). Null or nullifying players: the difference
between the Shapley value and equal division solutions. Journal of Economic
Theory 136, 767-775.\newline
\noindent Winter E (1989). A value for cooperative games with levels structure of cooperation. International Journal of Game Theory 18, 227-240.\newline

\end{document}